\documentclass[11pt]{article}

\usepackage{amsmath, amssymb, amsfonts}
\usepackage{amsopn}
\usepackage{natbib}

\usepackage{enumitem}
\setlist[enumerate]{leftmargin=.5in}
\setlist[itemize]{leftmargin=.5in}

\newcommand{\B}{\mathcal{B}}
\newcommand{\D}{\mathcal{D}}

\newcommand{\f}{\varphi}

\newcommand{\la}{\lambda}
\newcommand{\La}{\Lambda}
\newcommand{\M}{\mathcal{M}}

\newcommand{\R}{\mathbb{R}}

\DeclareMathOperator{\essinf}{essinf}
\DeclareMathOperator{\esssup}{esssup}

\newtheorem{theorem}{Theorem}

\newtheorem{corollary}[theorem]{Corollary}

\newtheorem{definition}[theorem]{Definition}
\newtheorem{example}[theorem]{Example}

\newtheorem{proposition}[theorem]{Proposition}

\newenvironment{proof}[1][Proof]{\noindent\textbf{#1.} }{\ \rule{0.5em}{0.5em}}

\begin{document}

\title{An axiomatization of $\La$-quantiles}
\author{Fabio Bellini \footnote{Department of Statistics and Quantitative Methods, University of Milano-Bicocca, \texttt{fabio.bellini@unimib.it}} \vspace{0.3cm}, 
Ilaria Peri\footnote{Corresponding author, Department of Economics, Mathematics and Statistics, Birkbeck, University of London, \texttt{i.peri@bbk.ac.uk}}}
\date{\today }
\maketitle

\begin{abstract}
We give an axiomatic foundation to $\La$-quantiles, a family of generalized quantiles introduced by \cite{FMP2014} under the name of Lambda Value at Risk. Under mild assumptions, we show that these functionals are characterized by a property that we call ``locality'', that means that any change in the distribution of the probability mass that arises entirely above or below the value of the $\La$-quantile does not modify its value. 
We compare with a related axiomatization of the usual quantiles given by \cite{C2009}, based on the stronger property of ``ordinal covariance'', that means that quantiles are covariant with respect to increasing transformations. Further, we present a systematic treatment of the properties of $\La$-quantiles, refining some of the results of \cite{FMP2014} and \cite{BPR2017} and showing that in the case of a nonincreasing $\La$ the properties of $\La$-quantiles closely resemble those of the usual quantiles. 
\\[2mm]



\end{abstract}

\section{Introduction}

Quantiles are perhaps the most basic type of statistical functional. Their widespread use is motivated by several reasons: they enjoy a number of nice theoretical properties, are well defined for any probability distribution, completely characterize the distribution itself, are robust with respect to outliers, and have a very simple probabilistic interpretation. 
Despite the well known drawbacks of Value at Risk such as its lack of subadditivity, that originated the stream of literature now known as axiomatic theory of risk measures, it is fair to say that quantiles do still play a central role both in practical and in mathematical finance.
Recently, \cite{FMP2014} introduced an interesting generalization of the very same notion of quantile, which maintains the basic structure of the usual definition but adds more flexibility in allowing a dependence of the probability level on the variable value; the resulting functional has been called Lambda Value at Risk and is defined as follows:
\[
\Lambda V@R(F):= -\inf \{x \in \mathbb{R}\,|\, F(x) >\Lambda(x)\},
\]
where $F$ is a distribution function and $\Lambda \colon \mathbb{R}\rightarrow (0,1)$ is a suitable monotone and right continuous function that is simply a constant in the case of usual quantiles. In \cite{BPR2017}, a first theoretical study of $\Lambda V@R$ has been provided, mainly focusing on the conditions under which it is robust, elicitable and consistent in the sense of \cite{D2016}. Further, $\Lambda V@R$ has been applied in financial risk management as an alternative to $V@R$ in order to assess capital requirements in \cite{CP2018} and \cite{HMP2018}. 

In this paper we first revisit the possible definitions and properties of $\La$-quantiles for a general function $\La$. We then focus on the case of a nonincreasing $\La$; under this assumption, we obtain in Proposition \ref{prop:properties_2}, Proposition \ref{prop:continuity} and  Proposition \ref{prop:continuity_2} several weak semicontinuity properties and mixture-convexity properties that closely resemble those of the usual quantiles. 
Intuitively, this is due to the fact that in this case the corresponding $\La$-quantiles can be visualized as intersection points or intersection intervals between $F$ and $\La$, exactly as the usual quantiles. Further, in the nonincreasing case the roles played by $F$ and $\La$ are essentially exchangeable, thus revealing a symmetry that is hidden in the definition of the usual quantiles. 

The financial interpretation of the assumption of a nonincreasing $\La$ will be discussed in Section \ref{sec:first} by means of a simple ``two-level'' Value at Risk example (Example 2.7).

The main theoretical contribution of the paper is an axiomatic foundation of $\Lambda$-quantiles given in Theorem \ref{th:axiomatization}. We pinpoint and formalize a crucial property, that we call ``locality'', that added to monotonicity, normalization and weak semicontinuity actually characterizes $\La$-quantile with a nonincreasing $\La$. We borrowed the name ``locality'' from an effective qualitative description of it that has been given by \cite{K2017} in the nice review about quantiles: 

\smallskip
``[...]\,quantiles are inherently \emph{local} and are nearly impervious to small perturbations of distributional mass. We can move mass around above and below the median without disturbing it at all provided, of course, that mass is not transferred from above the median to below or vice versa. This locality of the quantiles is highly advantageous for the same reasons that locally supported basis functions are advantageous in nonparametric regression: because it assures a form of robustness that is lacking in many conventional statistical procedures, notably those based on minimizing sums of squared residuals.\,[...]''.
\smallskip

The meaning of locality is obvious from the statistician's point of view but to the best of our knowledge has never been precisely formalized from a mathematical point of view. We propose the following definition: a statistical functional $T$ is local
if for any interval $(x,y) \subset \R$, 
\[
T(F)\in (x,y )\Rightarrow T(F)=T(G),\mbox{ for each } G \in \M \mbox{ with }%
G=F\mbox{ on }(x,y).
\]
This definition indeed captures the idea that modifying the distribution function $F$ in a way that leaves it unchanged in an arbitrarily small neighbourhood of $T(F)$ does not modify the value of $T$. In \cite{K2017}, the author was speaking about locality as a property of the usual quantiles; surprisingly, we show that this quite strong property holds also for more general functionals, and actually we prove in Theorem \ref{th:axiomatization} that under mild additional assumptions it axiomatizes $\La$-quantiles with a nonincreasing $\La$. 

It is also interesting to compare with a similar-in-spirit axiomatization of the usual quantiles given by \cite{C2009}, that is based on a proper stronger than locality called ordinal covariance, that requires that for all continuous bijection $\f \colon \R \to \R$ it holds that
\[
T(F \circ \f^{-1})= \f (T(F)).
\] 
The paper is structured as follows: Section 2 introduces the basic definitions of $\La$-quantiles and their properties, Section 3 presents the main axiomatization theorem, Section 4 concludes. All proofs are moved to the Appendix.

\section{Definitions and first properties of $\La$-quantiles}
\label{sec:first}

We denote respectively with $\M$ and $\M_c$ the set of probability measures on $(\R, \B(\R))$ and its subset of probability measures with compact support.  
Each $\mu \in \M$ is identified with its right continuous distribution function $F_\mu(x)=\mu(-\infty,x]$.
A functional $T \colon \M \to \R$ is said to be normalized if $T(\delta_x) =x$ for each $x \in \R$ and monotonic if $F_1 \geq_{st} F_2 \Rightarrow T(F_1) \geq T(F_2)$, where the usual stochastic order $\geq_{st}$ is defined by $
F_1 \geq_{st}F_2$ if $F_1 (x) \leq F_2(x)$, for each $x \in \R$.
We recall the definition of the left and right quantiles of a distribution function.
\begin{definition}\label{def:quantiles}
Let $F \in \M$ and $\lambda \in [0,1]$. Then
\begin{align*}
q_\lambda^-(F) &:=\inf \left\{ x\in \R\,|\,F(x)\geq \lambda \right\}= \sup \left\{ x\in \R\,|\,F(x)<\lambda \right\},\\
q_\lambda^+(F) &:=\inf \left\{ x\in \R\,|\,F(x)>\lambda \right\} =\sup \left\{ x\in \R\,|\,F(x) \leq \lambda \right\}.
\end{align*}
\end{definition}
As usual, we set by convention $\inf(\emptyset)=+\infty 
$ and $\sup (\emptyset)=-\infty$, so $q_0^-(F)=-\infty$, $q_1^+(F)=+\infty$, while $q_0^+(F)=\essinf(F)$, $q_1^-(F)=\esssup(F)$, where $\essinf(F)$ and $\esssup(F)$ denote respectively the essential infimum and the essential supremum of $F$.  

An interesting generalization of Definition \ref{def:quantiles} arises by replacing the constant $\lambda \in [0,1]$ with a variable threshold $\La \colon \R \to [0,1]$, called ``probability-loss function''
in \cite{FMP2014}, where the resulting functional (up to a sign change, and with minor assumptions on $\La$) has been called Lambda Value at Risk. To begin with a setting as general as possible, we start considering four possible definitions. 

\begin{definition}\label{def:Lambda}
Let $F \in \M$ and $\La \colon \R
\rightarrow [0,1]$. The $\La$-quantiles of $F$ are defined as follows:
\begin{align*}
q_\La^-(F) &:=\inf \left\{ x\in \R\,|\,F(x)\geq \La
(x)\right\},
&&q_\La^+(F) :=\inf \left\{ x\in \R\,|\,F(x)>\La (x)\right\}, \\
\tilde{q}_\La^-(F) &:=\sup \left\{ x\in \R\,|\,F(x)<\La
(x)\right\},
&&\tilde{q}_\La^+(F) :=\sup \left\{ x\in \R\,|\,F(x)\leq
\La (x)\right\}.
\end{align*}
\end{definition}

From a financial point of view, as discussed in \cite{FMP2014}, \cite{CP2018} and \cite{HMP2018}, it is sensible allowing a dependence of the ``acceptability'' probability level $\lambda$ on the corresponding amount of the loss $x$. 
A related idea has been investigated in \cite{BBM2019}, that constructed translation invariant risk measures based on benchmark loss distributions. 
At this level of generality, $\La$-quantiles satisfy only a few properties that are collected in the following proposition. 
Recall that $T \colon \M \to \R$ is said to be m-quasiconvex if for each $F_1, F_2 \in \M$ and $\alpha \in (0,1)$ it holds that 
\[
T(\alpha F_1 + (1-\alpha)F_2) \leq \max(T(F_1), T(F_2)),
\]
that is if for each $\gamma \in \R$ the lower level sets of $T$ are convex with respect to mixtures, and similarly $T$ is said to be m-quasiconcave if 
\[
T(\alpha F_1 + (1-\alpha)F_2) \geq \min(T(F_1), T(F_2)).
\]

\begin{proposition}\label{prop:properties}
Let  $q_\La ^-$, $q_\La^+$, $\tilde{q}_\La^-$, $\tilde{q}_\La^+$ be as in
Definition \ref{def:Lambda}. Then:
\begin{itemize}
\item [a)]  $q_\La^-(F)\leq q_\La^+(F)$ and $\tilde{q}_{\La}^{-}(F)\leq \tilde{q}_\La^+(F)$
\item [b)] $q_\La ^-$, $q_\La^+$, $\tilde{q}_\La^-$, $\tilde{q}_\La^+$ are monotonic 
\item [c)] $q_\La ^-$, $q_\La^+$, $\tilde{q}_\La^-$, $\tilde{q}_\La^+$ are monotonic with respect to $\La$, in the sense that if $\La_1 \leq \La_2$ then each of the four $\La_1$-quantiles is smaller than the corresponding $\La_2$-quantile
\item [d)] ${q}_\La^-$ and ${q}_\La^+$ are m-quasiconcave, $\tilde{q}_\La^-$ and $\tilde{q}_\La^+$ are m-quasiconvex. 
\end{itemize}
\end{proposition}

Without additional assumptions the four $\La$-quantiles can be all different, not finite, and even fail to satisfy the normalization property, as the following simple examples show. 
\begin{example}\label{ex_1}
Let $\La(x)=1_{\{x>0\}}$ and let
$F = \frac{1}{2} \delta_{-1} + \frac{1}{2}\delta_1$.
Then 
\[
q_\La^-(F)=-\infty, \, q_\La^+(F)=-1, \, \tilde{q}_\La^-(F)=1, \, \tilde{q}_\La^+(F)=+\infty.
\]
\end{example}
\begin{example}\label{ex_2}
Let $\La(x)=1_{\{x<0\}}$ and let $F=\delta_1$, $G=\delta_{-1}$. Then 
\begin{align*}
&q_\La^-(F)= 0, \, q_\La^+(F)= 1, \, \tilde{q}_\La^-(F)= 0, \, \tilde{q}_\La^+(F)= 1,\\
&q_\La^-(G)= -1, \, q_\La^+(G)= 0, \, \tilde{q}_\La^-(G)= -1, \, \tilde{q}_\La^+(G)= 0.
\end{align*}
\end{example}

As anticipated in the Introduction, an assumption that maintains a substantial level of generality with respect to the usual quantiles but at the same time avoids pathological situations is that the
function $\La$ is nonincreasing. 
It has to be remarked that 
\cite{FMP2014} and \cite{BPR2017} obtained a few results on $\La$-quantiles also under alternative assumptions on $\La$. 
We collect in the following proposition the properties of $\La$-quantiles in the nonincreasing case.
Recall that $T \colon \M \to \R$ is said to have the CxLS property if for each $F_1, F_2 \in \M$, $\alpha \in (0,1)$ and $\gamma \in \R$ it holds that 
\[
T(F_1) = T(F_2) = \gamma \Rightarrow T(\alpha F_1 + (1- \alpha) F_2) = \gamma, 
\]
that is if the level sets of $T$ are convex with respect to mixtures.

\begin{proposition}\label{prop:properties_2}
Let $q_\La ^-$, $q_\La^+$, $\tilde{q}_\La^-$, $\tilde{q}_\La^+$ be as in
Definition \ref{def:Lambda} and let $\La \colon \R \to [0,1]$ be nonincreasing. Then:
\begin{itemize}
\item [a)] $\tilde{q}_\La^-(F)=q_\La^-(F)$ and $\tilde{q}_\La^+(F)=q_\La^+(F)$
\item [b)] $q_\La^-$ and $q_\La^+$ are finite if and only if $\La \not \equiv 0$ and $\La \not \equiv 1$
\item [c)] if $\La \not \equiv 0$ and $\La \not \equiv 1$, then $q_\La^-$ is normalized if and only if $\La(x) >0$ for each $x \in \R$, and
$q_\La^+$ is normalized if and only if $\La(x) <1$ for each $x \in \R$. 
\item [d)] If $\La_1(x)=\La_2(x)$ on their common points of continuity, then $q_{\La_1}^-(F)=q_{\La_2}^-(F)$ and $q_{\La_1}^+(F)=q_{\La_2}^+(F)$.
\item [e)] $q_\La^-(F)$ and $q_\La^+(F)$ have the CxLS property. 
\end{itemize}
\end{proposition}
\noindent
Thus in the nonincreasing case, $\La$-quantiles reduce from four to two, are finite and satisfy the normalization property if $0 < \La < 1$. Looking again at Example \ref{ex_2}, we see that indeed $\tilde{q}_\La^-(F)=q_\La^-(F)= 0$ and  $\tilde{q}_\La^+(F)=q_\La^+(F)=1$, and the normalization property fails since $\La$ takes the values $0$ and $1$. Item d) implies that in the nonincreasing case it is always possible to assume w.l.o.g. that $\La$ is either right or left continuous. Concerning e), we recall that
in \cite{BPR2017} it was shown that the CxLS property may fail for an increasing $\La$, although it is satisfied for increasing distributions if $\La$ is piecewise constant and has a finite number of jumps. \\
Perhaps the simplest nontrivial example of a $\La$ function satisfying the assumptions of Proposition \ref{prop:properties_2} is the following ``two-level'' or ``double'' quantile.
\begin{example}
Let $0 < \alpha < \beta <1$, $\bar{x} \in \R$ and $\La(x)=\beta 1_{\{x \leq \bar{x}\}} + \alpha 1_{\{x > \bar{x}\}}$. Then
\begin{equation}\label{two_level}
q_\La^-(F)=
\begin{cases}
q_\beta^-(F) &\text { if } q_\beta^-(F) \leq \bar{x}\\
q_\alpha^-(F) &\text { if } q_\alpha^-(F) \geq \bar{x}\\
\bar{x} &\text { if } q_\alpha^-(F) < \bar{x} < q_\beta^-(F).
\end{cases}
\end{equation}
\end{example}
This example has a simple financial interpretation that well illustrates the potential of $\La$-quantiles. Assume that the distribution function $F$ refers to future losses and that $\alpha=0.95$ and $\beta=0.99$. We interpret $\bar{x}$ as a reference level of the capital requirement that the risk manager (RM/she henceforth) has already established ex ante. This might come for example as an output of a simpler model, from a regulatory requirement, from her expert judgement or past experience. As always, the RM is pursuing contrasting objectives: to be conservative, but not too much, in order to avoid immobilizing an excessive amount of capital. A straightforward line of reasoning leading to \eqref{two_level} is the following:
if $q_{0.95} \geq \bar{x}$ then $q_{\La}=q_{0.95}$ and the RM is satisfied with a $95\%$-quantile that is enough conservative with respect to her reference level $\bar{x}$. If instead $q_{0.95} < \bar{x}$, the RM is worried to be not conservative enough and either directly sets $q_{\La}=\bar{x}$, or switches to the $99\%$-quantile if setting $q_{\La}=\bar{x}$ seems too conservative to her because $\bar{x} > q_{0.99}$. Summing up, moving from the usual quantiles $q_{0.95}$ or $q_{0.99}$ to $q_{\La}$ given by \eqref{two_level} allows the RM to keep into account also her a priori calculation of the capital requirement $\bar{x}$.

We end the section by discussing the weak semicontinuity properties of $\La$-quantiles, complementing a few results of \cite{FMP2014} and \cite{BPR2017}. 
Recall first that $F_n \overset{d}{\to} F$ if $F_n(x) \to F(x)$ in the continuity points of $F$, and $T \colon \M \to \R$ is said to be respectively
weakly continuous at $F$ if $F_n \overset{d}{\to} F \Rightarrow T(F_n) \to T(F)$,
weakly lower semicontinuous at $F$ if
$F_n \overset{d}{\to} F \Rightarrow T(F) \leq \liminf_{n \rightarrow +\infty} T(F_n)$, and
weakly upper semicontinuous at $F$ if $F_n \overset{d}{\to} F \Rightarrow T(F) \geq \limsup_{n \rightarrow +\infty} T(F_n)$. \\
In \cite{FMP2014} it was shown that  if $\La$ is right continuous with $\sup_{x \in \R} \La(x) < 1$, then $q_\La^+$ is weakly upper semicontinuous. Notice that \cite{FMP2014} adopted a slightly different definition of $q_\La^+$, that is equivalent to Definition \ref{def:Lambda} in the case of a nonincreasing $\La$, but not equivalent in more general cases.  
In \cite{BPR2017} it was shown that 
if $\La$ is continuous, monotone (either nonincreasing or nondecreasing) and satisfies $\inf_{x \in \R} \La(x) > 0, \sup_{x \in \R} \La(x) < 1$,
then 
$q_\La^+$ is weakly continuous on the set of distribution functions that do not coincide with $\La$ on any interval. 

In the next proposition we show that in the nonincreasing case the weak continuity properties of $\La$-quantiles closely resemble those of the usual quantiles, in the sense that the left $\La$-quantile is  weakly lower semicontinuous and the right $\La$-quantile is weakly upper semicontinuous. The proof is based on a variant of Proposition 2.5 from \cite{FMP2014}, that is recalled in the Appendix.

\begin{proposition}\label{prop:continuity}
Let $\La \colon \R\rightarrow [0,1]$ be nonincreasing. Then
$q_\La^-$ is weakly lower semicontinuous and $q_\La^+$ is weakly upper semicontinuous. 
\end{proposition}
As an immediate consequence we have the following.
\begin{corollary}\label{cor:continuity}
Let $\La \colon \R\rightarrow [0,1]$ be nonincreasing. Then:
\begin{itemize}
\item [a)] if $q_\La^-(F) = q_\La^+(F)$ then
$q_\La^-(F)$ and $q_\La^+$ are weakly continuous at $F$
\item [b)] if $F$ is increasing then $q_\La^-$ and $q_\La^+$ are weakly continuous at $F$
\item [c)] if $\La$ is decreasing then $q_\La^-$ and $q_\La^+$ are weakly continuous at each $F \in \M$.
\end{itemize}
\end{corollary}

\noindent
As anticipated in the Introduction, an interesting aspect of the definition of $\La$-quantiles is that it reveals a symmetry between the roles played by $F$ and $\La$, that is hidden in the definition of the usual quantiles. We have seen a first instance of this symmetry in the monotonicity property of Proposition \ref{prop:properties}, that holds either with respect to $F$ or with respect to $\La$. 
Here, we have another example of this symmetry in the following continuity result with respect to $\La$, that holds when $F$ is fixed and $\La_n \to \La$ in the sense specified by the proposition below.

\begin{proposition}\label{prop:continuity_2}
Let $\La_n, \La \colon \R\rightarrow [0,1]$ be nonincreasing. Then:
\begin{align*}
&\La _{n}(x)\uparrow \La (x) \text { in the continuity points of }%
\La \Rightarrow q_{\La _{n}}^{-}(F)\rightarrow q_\La^-(F)\\
&\La _{n}(x)\downarrow \La (x) \text{ in the continuity points of }%
\La \Rightarrow q_{\La _{n}}^{+}(F)\rightarrow q_\La^+(F).
\end{align*}
As a corollary, if $F$ is increasing then 
\[
\La _{n}(x) \to \La (x) \text{ in the continuity points of } \La \Rightarrow q_{\La_n}(F) \to q_\La(F).
\]
\end{proposition}

\noindent
The proof is similar to the one of Proposition \ref{prop:continuity} and omitted. 

\section{Axiomatization of $\La$-quantiles}
\label{sec:axiomatization}

In this section we provide an axiomatization of $\La$-quantiles with a nonincreasing $\La$. As anticipated in the Introduction, 
the crucial property that axiomatizes $\La$-quantiles can be intuitively described as follows: however the
distribution of the probability mass is modified, either on the left or on
the right of the value of the functional, the functional itself remains unchanged. Formally, we give the following. 

\begin{definition} \label{def:locality}
Let $T\colon \M \to \R$. We say that $T$ is
\emph{local} if for any $(x,y) \subset \R$, 
\[
T(F)\in (x,y )\Rightarrow T(F)=T(G),\mbox{ for each } G \in \M \mbox{ with }%
G=F\mbox{ on }(x,y).
\]
\end{definition}
In other words, the value of the functional $T(F)$ is determined by the values of the distribution function $F$ in arbitrarily small open neighbourhoods $(x,y)$ of the value of the functional itself.
The usual quantiles are local, and from Definition \ref{def:Lambda} it follows immediately that more generally also $\La$-quantiles with a nonincreasing $\La$ are local. On the contrary, $\La$-quantiles with an increasing $\La$ do not satisfy the locality property, as it can be seen in the following example. 
\begin{example}
Let $\La(x) = \Phi(x)$, the distribution function of a standard normal, and let $F(x)=\Phi(2x)$.\ Then $q_\La^-(F)=\inf \left\{ x\in \R\,|\,F(x)\geq \La
(x)\right\}=0$, that is the only intersection point between $F$ and $\La$. Letting $G$ be any distribution such that $G(x)=F(x)$ for $x \in (-1, +\infty)$ and $G(\bar{x})\geq \La(\bar{x})$ for some $\bar{x} < -1$, it follows that 
$q_\La^-(G)<\bar{x}<0$, thus violating the locality. 
\end{example}

The following theorem is the main result of the paper and shows that under mild assumptions a local functional is a $\La$-quantile with a nonincreasing $\La$. 	

\begin{theorem}\label{th:axiomatization}
Let $T \colon \M \to \R$ be normalized, monotonic, local and weakly lower semicontinuous. Then there exists a nonincreasing $\La \colon \R \to [0,1]$ such that $T(F)=q_{\La
}^{-}(F)$. Similarly, if under the same assumptions $T$ is weakly upper semicontinuous, then there exists a nonincreasing $\La \colon \R \to [0,1]$ such that $T(F)=q_{\La
}^{+}(F)$.
\end{theorem}

The proof is postponed to the Appendix. It is very interesting to compare our result with an axiomatization of the usual quantiles given by \cite{C2009}, based on the property of ordinal covariance, that we formally recall below. 
\begin{definition}\label{def:ord_cov}
A functional $T \colon \M_c \to \R$ satisfies the ordinal covariance property if for all strictly increasing and continuous $\f \colon \R \to \R$ satisfying $\f(\R)=\R$ it holds that
\begin{equation*}
T(F \circ \f^{-1})= \f (T(F)).
\end{equation*}
\end{definition}
Ordinal covariance is a very strong property that implies normalization, translation invariance and positive homogeneity. Chambers proved the following:
\begin{theorem}\label{th:Chambers_1}
Let $T \colon \M_c \to \R$ be monotonic and satisfy the ordinal covariance property. 
If $T$ is weakly upper semicontinuous, then there exists $\la \in [0,1)$ such that $T(F)=q_\lambda^+(F)$, while if $T$ is weakly lower semicontinuos then there exist $\la \in (0,1]$ such that $T(F)=q_\lambda^-(F)$.
\end{theorem}
\noindent
In order to cover also the case of functionals defined on distributions with unbounded support, Chambers introduced a slightly stronger notion of ordinal covariance, that involves also possibly unbounded transformations; we refer again to \cite{C2009} for the details. \\

The first step of Chambers' proof is to show that on a dyadic distribution of the form
\[
B_{x,y}^\la=\la \delta_x + (1-\la)\delta_y, \quad \text{ with }\la \in (0,1),
\]
an ordinal covariant functional $T$ is always \emph{boundary}, that is $T(B_{x,y}^\la)=x$ or $T(B_{x,y})=y$, and the same applies to any finitely supported distribution $F$, in the sense that the possible values of $T$ are only the point masses of $F$.~In contrast, under our weaker assumption of locality, the functional $T$ can also be \emph{internal} on dyadic variables, in the sense that it can happen that $T(B_{x,y}^\la) \in (x,y)$.~However, as we show in the Appendix, a crucial consequence of locality is that for each $\la \in (0,1)$ there can be only one possible internal value of $T$ on dyadic distributions, that we will denote with $z(\la)$. This will allow us to show that on the set of dyadic distributions the functional $T$ is a $\La$-quantile for a suitable nonincreasing $\La$, determined as a generalized inverse of $z(\la)$.~Finally, from locality and weak semicontinuity we will show that $T$ can be uniquely extended from dyadic distributions to the whole of $\M$.

\section{Conclusions and directions for further research}

In this short note we have given a comprehensive study of the properties of $\La$-quantiles in the case of a nonincreasing $\La$. Such functionals enjoy a few properties that are similar to the one of the usual quantiles, and can be axiomatized as in Theorem  \ref{th:axiomatization} by means of the locality property formalized in Definition \ref{def:locality}, that has been discussed in the statistical literature only on a purely qualitative basis. 

From a mathematical point of view, we underline at least two open problems: the first is the axiomatization of $\La$-quantiles for a general shape of the function $\La$, or under alternative assumptions (for example, in the nondecreasing case); the second is the characterization of $\La$-quantiles that have the CxLS property. We know that this is the case if $\La$ is nonincreasing and in a few other situations studied in \cite{BPR2017}, but a complete characterization is still lacking. 

From the point of view of financial and statistical applications, $\La$-quantiles with a nonincreasing $\La$ are elicitable, as noticed in \cite{BB2015} and further accurately discussed in \cite{BPR2017}, that is consistent scoring functions do exist.
This opens the way for the development of statistical procedures that generalize quantile regression and are the subject of current investigation.

\section{Appendix}


\begin{proof}[Proof of Proposition \ref{prop:properties}] 
a), b) and c) follow immediately from Definition \ref{def:Lambda}.
To show the first line of (d), notice that the thesis is trivially satisfied if 
$q_\La^-(F_1) = -\infty$ or $q_\La^-(F_2) = -\infty$. Otherwise,
let $\bar{z}=\min (q_\La^-(F_1),q_\La^-(F_2))$. 
From the definition of $q_\La^-$, it follows that for each $z < \bar{z}$ it holds that
$F_1(z)<\La (z)$ and $F_2(z)<\La (z)$, so $\lambda F_1(z)+(1-\lambda
)F_2(z)<\La (z)$. Again from the definition of $q_\La^-$, it follows that $q_\La^-(\lambda F+(1-\lambda )G)\geq \bar{z}$. The other cases of (d) follow similarly.
\end{proof}

\begin{proof}[Proof of Proposition \ref{prop:properties_2}]
a) If $\La $ is nonincreasing then 
$\left\{ F(x)<\La (x)\right\}$, $\left\{ F(x)\geq \La (x)\right\}$,
$\left\{ F(x)\leq \La
(x)\right\}$ and $\left\{ F(x)>\La (x)\right\} $, are (possibly empty) disjoint intervals whose union is $\R$, from which the thesis follows.
b) If $\La \equiv 0$ or $\La \equiv 1$ then $q_\La^-$ and $q_\La^+$ are not finite as we remarked after Definition \ref{def:quantiles}.
Conversely, if $\La \colon \R \to [0,1]$ is not identically equal to $0$ or $1$ then $\lim_{x\to
-\infty }\La (x)>0$ and $\lim_{x\to +\infty }\La (x)<1$, so there exists $x_{1}$ and $x_{2}$ such that $F(x)<\La
(x)$ for $x<x_1$ and $F(x)>\La(x)$ for $x>x_2$, from which it follows that
$x_1 \leq q_\La^-(F) \leq q_\La^+(F) \leq x_2$.
d) Assume by contradiction and without loss of generality that $q_{\La_1}^-(F)<q_{\La_2}^-(F)$. For each $z \in (q_{\La_1}^-(F), q_{\La_2}^-(F))$ it follows that $\La_1(z) \leq F(z) < \La_2(z)$, that gives a contradiction since the set $\{z \in \R \, | \, \La_1(z) = \La_2(z)\}$ is dense in $\R$. The proof for $q_\la^+$ is similar. 
e) follows from a) and Proposition \ref{prop:properties} item d). 
\end{proof}

In order to prove Proposition \ref{prop:continuity}, we first recall the following characterization of weak semicontinuity given in Proposition 2.5 from \cite{FMP2014}, slightly modified for our sign conventions.  
\begin{proposition}\label{lemma:FMP}
Let $T\colon \M\rightarrow \R$ be monotonic. 
Then $T$ is weakly lower semicontinuous if and only if 
\[
F_{n}(x)\downarrow F(x)\text{ in the continuity points of }F\Rightarrow
T(F_{n})\rightarrow T(F), 
\]%
and $T$ is weakly upper semicontinuous if and only if 
\[
F_{n}(x)\uparrow F(x)\text{ in the continuity points of }F\Rightarrow
T(F_{n})\rightarrow T(F).
\]%
\end{proposition}

\begin{proof}[Proof of Proposition \ref{prop:continuity}]
Let us consider the case of $q_\La^-$. From Proposition \ref{lemma:FMP} it is enough to show that if $F_n \downarrow F$ in the continuity points of $F$, then $q_\La^-(F_n) \to q_\La^-(F)$. 
Let $z_n:=q_\La^-(F_n)$. From the monotonicity of $q_\La^-$ it follows that 
$z_n \uparrow z \leq q_\La^-(F)$. Assume by contradiction that $z < q_\La^-(F)$. Since the set of continuity points of $F$ is dense in $\R$, there exists $\bar{z}$, a continuity point of $F$, such that $z_n \leq z < \bar{z} < q_\La^-(F)$. 
From the definition of $q_\La^-(F)$ it follows that $F_n(\bar{z}) \geq \La(\bar{z})$ and $F(\bar{z}) < \La(\bar{z})$. Letting $n \to +\infty$ gives
$F (\bar{z}) \geq \La(\bar{z})$ and $F(\bar{z}) < \La(\bar{z})$, that is a contradiction. The proof for $q_\La^+$ is similar. 
\end{proof}

\begin{proof}[Proof of Corollary \ref{cor:continuity}] 
a) Let $F_n \overset{d}{\to} F$. From the weak lower semicontinuity of $q_\La^-$, the weak upper semicontinuity of $q_\La^+$ and Proposition \ref{prop:properties}, it follows that
\begin{align*}
q_\La^-(F) \leq \liminf_{n \to +\infty} q_\La^-(F_n) \leq 
\limsup_{n \to +\infty} q_\La^-(F_n) \leq \limsup_{n \to +\infty} q_\La^+(F_n) \leq 
q_\La^+(F).
\end{align*}
Since $q_\La^-(F) = q_\La^+(F)$ the inequalities are indeed equalities, that proves weak continuity of $q_\La^-$. The proof for $q_\La^+$ is similar. b) and c) follow from a) and from the uniqueness of $\La$-quantiles when either $F$ is increasing or $\La$ is decreasing. 
\end{proof}

\begin{proof}[Proof of Theorem \ref{th:axiomatization}]
The first step of the proof is to analyze the behaviour of $T$
on the set $\B$ of dyadic distributions 
of the form 
\[
B_{x,y}^\la:=\la \delta_x + (1-\la)\delta_y,
\]
for $x < y$ and $\la \in (0,1)$. 
From monotonicity and normalization it follows that 
\[
x \leq T(B_{x,y}^{\la}) \leq y.
\]
We say that $T$ is \emph{internal} on $B_{x,y}^{\la} \in \B$ if $T(B_{x,y}^{\la}) \in (x,y)$ and that $T$ is \emph{boundary} on $B_{x,y}^{\la}$ if $T(B_{x,y}^{\la}) =x $ or $T(B_{x,y}^{\la}) =y$. 
For each fixed $\la \in (0,1)$, there are two alternatives:

\begin{itemize}
\item [i)] there exists $x,y$ such that $T$ is internal on $B_{x,y}^{\la}$
\item [ii)] for each $x,y$, $T$ is boundary on $B_{x,y}^{\la}$. 
\end{itemize}
\noindent
Clearly, i) and ii) are mutually exclusive; we denote with $I_1$ the set of $\la \in (0,1)$ such that i) holds and with
$I_2$ the set of $\la$ such that ii) holds. Notice that either $I_1$ or $I_2$ can be empty; for example in the case of the usual quantiles $I_1=\emptyset$. We will discuss the various cases at a later stage of the proof. 

Assume that $\la \in I_1$ and let $\bar{z}:=T(B_{x,y}^{\la})$, so $\bar{z} \in (x,y)$. We first show that such a $\bar{z}$ is uniquely determined by $\la$ and does not depend on a particular choice of $x$ and $y$. Indeed, assume by contradiction that there exists $x_1,x_2,y_1,y_2$ such that 
\begin{align*}
T(B_{x_1,y_1}^{\la}) &= \bar{z}_1 \in (x_1, y_1), \\
T(B_{x_2,y_2}^{\la}) &= \bar{z}_2 \in (x_2, y_2), 
\end{align*}
with $\bar{z}_1 \neq \bar{z}_2$. From locality, it can be assumed w.l.o.g. that $x_1 < \bar{z}_1 < y_1 < x_2 <  \bar{z}_2 < y_2$. But then again from locality it would follow that $T(B_{x_1,y_2}^{\la}) =T(B_{x_1,y_1}^{\la})=\bar{z}_1$ and also $T(B_{x_1,y_2}^{\la}) = T(B_{x_2,y_2}^{\la})=\bar{z}_2$, that gives a contradiction. Hence, to each $\la \in I_1$ it is possible to associate a unique internal value $\bar{z}$ of the functional $T$ on $B_{x,y}^{\la}$, that we will denote with a slight abuse of notation by $\bar{z}(\la)$. Notice also that if $\la \in I_1$ and $y \leq z(\la)$, then again by locality it must hold that $T(B_{x,y}^\la)=y$, and similary if $x \geq z(\la)$ then $T(B_{x,y}^\la)=x$. Thus on dyadic variables with $\la \in I_1$, the function $\bar{z}$ completely determines $T$. 

We now show that on $I_1$, $\bar{z}(\la)$ is a nonincreasing function of $\la$. Indeed, assume by contradiction that there exists $\la_1 < \la_2$ such that $T(B_{x_1,y_1}^{\la_1}) = \bar{z}_1 \in (x_1,y_1)$ and 
$T(B_{x_2,y_2}^{\la_2}) = \bar{z}_2 \in (x_2,y_2)$, with $\bar{z}_1 < \bar{z}_2$. As before, from locality we can assume w.l.o.g. that $x_1 < \bar{z}_1 < y_1 < x_2 <  \bar{z}_2 < y_2$, and it would follow that 
$T(B_{x_1,y_2}^{\la_1}) = \bar{z}_1$ and $T(B_{x_1,y_2}^{\la_2}) = \bar{z}_2$, but since
\[
B_{x_1,y_2}^{\la_1} \geq_{st} B_{x_1,y_2}^{\la_2},
\]
from monotonicity we would get that $\bar{z}_1 \geq \bar{z}_2$, that gives a contradiction. 

We now show that the set $I_1$ is an interval. Let $\la_1, \la_2 \in I_1$ with $\la_1 < \la_2$, so there exists $x_1,y_1$, $x_2,y_2$ such that $T$ is internal on $B_{x_1,y_1}^{\la_1}$ and on $B_{x_2,y_2}^{\la_2}$. Letting $\bar{z}_1=T(B_{x_1,y_1}^{\la_1})$ and $\bar{z}_2=T(B_{x_2,y_2}^{\la_2})$, from the monotonicity of $\bar{z}(\la)$ on $I_1$ it follows that $\bar{z}_1 \geq \bar{z}_2$, so $x_2 < y_1$. For each $\la \in (\la_1, \la_2)$, it holds 
\[
B_{x_2,y_1}^{\la_1} \geq_{st} B_{x_2,y_1}^{\la} \geq_{st} B_{x_2,y_1}^{\la_2},
\]
hence from monotonicity
\[
y_1 > \bar{z}_1 = T(B_{x_2,y_1}^{\la_1}) \geq T(B_{x_2,y_1}^{\la}) \geq T(B_{x_2,y_1}^{\la_2})=\bar{z}_2 > x_2,
\]
and since that $T(B_{x_2,y_1}^\la) \in (x_2,y_1)$ it follows that $\la \in I_1$. 

Let us now consider the case $\lambda \in I_2$, in which by definition $T$ is boundary on each $B_{x,y}^\la$. \\
To ease the exposition, we say that $T$ is \emph{left} on $B_{x,y}^\la$ if $T(B_{x,y}^\la)=x$ and that $T$ is \emph{right} on $B_{x,y}^\la$ if $T(B_{x,y}^\la)=y$. For a fixed $\la \in (0,1)$, we first show that either $T$ is right for each $B_{x,y}^\la \in \B$ or $T$ is left for each $B_{x,y}^\la \in \B$. Indeed, assume by contradiction that there exists $x_1, y_1$, $x_2, y_2$ with $T(B_{x_1,y_1}^{\la}) = x_1$ and 
$T(B_{x_2,y_2}^{\la}) = y_2$. There are two possible cases: a) $x_1 <y_2$, b) $x_1 \geq y_2$. In case a), from locality it is possible to assume w.l.o.g. that $x_1 < y_1 < x_2 < y_2$, that gives a contradiction since $T(B_{x_1,y_2}^\la)$ should be at the same time equal to $x_1$ and to $y_2$. In case b), it follows that $x_2 < y_2 \leq x_1 < y_1$, and since 
\[
B_{x_1,y_1}^{\la}\geq_{st} B_{x_2,y_1}^{\la} \geq_{st} B_{x_2,y_2}^{\la},
\]
from monotononicity it follows that 
\[
y_1 > x_1=T(B_{x_1,y_1}^{\la})\geq T(B_{x_2,y_1}^{\la}) \geq T(B_{x_2,y_2}^{\la})=y_2 >x_2,
\]
so $T$ would be internal on $B_{x_2,y_1}^{\la}$, contradicting the assumption $\la \in I_2$. 

Let now $\la_1 \in I_1$ and let $\la_2 \in I_2$, with $\la_2 > \la_1$. We show that $T
$ must be left on each $B_{x,y}^{\la_2}$. Indeed, from the definition of $I_1$ there exists $x_1, y_1$ such that $T(B_{x_1,y_1}^{\la_1}) = z_1 \in (x_1, y_1)$,
and assuming by contradiction that $T$ is right on each $B_{x,y}^{\la_2}$ it would follow that 
$T(B_{x_1,y_1}^{\la_2}) = y_1$. But then monotonicity leads to a contradiction since 
\[
B_{x_1,y_1}^{\la_2} \leq_{st} B_{x_1,y_1}^{\la_1}
\]
and 
\[
T(B_{x_1,y_1}^{\la_2}) =y_1 > z_1 = T(B_{x_1,y_1}^{\la_1}).
\]
A similar argument shows that if  $\la_1 \in I_1$ and $\la_2 \in I_2$, with $\la_2 < \la_1$, then $T$ must be right on each $B_{x,y}^{\la_2}$. Further, the same argument shows also that if $\la_1 \in I_2$ with $T$ always left and $\la_2 > \la_1$, then $\la_2 \in I_2$ and $T$ is always left, while if $\la_1 \in I_2$ with $T$ always right and $\la_2 < \la_1$, then $\la_2 \in I_2$ and $T$ is always right. 

Denoting with $I_2^r$ and $I_2^\ell$ respectively the subintervals of $I_2$ such that $T$ is either always right or left, the previous analysis shows that there are the following seven possibilities:
a) $I=I_2^r$,  
b) $I=I_2^\ell$, 
c) $I=I_2^r \cup I_2^\ell$,
d) $I=I_1$,
e) $I=I_2^r  \cup I_1$,
f) $I=I_1 \cup I_2^\ell$,
g) $I=I_2^r  \cup I_1 \cup I_2^\ell$,
where the unions of intervals are ordered and disjoint. In all the above cases, the function $z \colon (0,1) \to \R \cup \{+\infty\} \cup \{-\infty\}$ defined by
\[
z(\la):=
\begin{cases}
+\infty &\text{ if } \la \in I_2^r\\
\bar{z}(\la) &\text{ if } \la \in I_1 \\
-\infty &\text{ if } \la \in I_2^\ell
\end{cases}
\]
is nonincreasing and finite if and only if $I_2=\emptyset$. Moreover, the previous analysis shows that the function $z$ completely determines $T$ on $\B$ in the sense that
\begin{align}
y \leq z(\la) &\Rightarrow T(B_{x,y}^\la)=y \notag \\
x \geq z(\la) &\Rightarrow T(B_{x,y}^\la)=x \label{cond} \\
x < z(\la) < y &\Rightarrow T(B_{x,y}^\la)=z(\la) \notag.
\end{align}

The second step of the proof is to show that the functional $T$ restricted on $\B$ is actually a $\La$-quantile for a suitable $\La$. For this it is necessary to use the weak semicontinuity assumption. 
We consider the case in which $T$ is weakly lower semicontinuous, the other case being similar. 
Notice that from Lemma \ref{lemma:FMP} it follows immediately that $z(\la)$ is right continuous. Define
\begin{equation}\label{LambdaT}
\La(x) := \inf \{ \la \in (0,1) \, | \, z(\la) \leq x \},
\end{equation} 
with the convention that $\inf \emptyset = 1$. Since $z(\la)$ is nonincreasing and right continuous, then 
\begin{equation}\label{conn}
\la < \La(x) \iff z(\la) > x \text { and } \la \geq \La(x) \iff z(\la) \leq x. 
\end{equation}
Clearly, $\La \colon \R \to [0,1]$ is nonincreasing. We show that 
\begin{equation} \label{qLaleft}
T(F)=q_\La^-(F), \text{ for each } F \in \B,
\end{equation} 
where $\La$ is defined in \eqref{LambdaT}.
To prove it we examine the cases a) - g) outlined before. In case a), $T$ is right on each dyadic variable, so $T(F)=\esssup(F)$. Since $z(\la)=+\infty$, from \eqref{LambdaT} we have $\La(x)=1$, and indeed $T(F)=q_1^-(F)$ on $\B$. Similarly, in case b) $T$ is always left, that would correspond to $T(F)=\essinf(F)$, but this case cannot occur because of the lower semicontinuity assumption. Case c) corresponds to the usual left quantiles, where 
\[
z(\la) =
\begin{cases}
+ \infty \text{ if } 0 < \la < \alpha \\
-\infty \text{ if } \alpha \leq \la <1
\end{cases}
\]
and from \eqref{LambdaT} it follows that $\La(x) = \alpha$, where $\alpha= \sup I_1 = \inf I_2$. In case d), we verify that $q_\La^-$, where $\La$ is given by \eqref{LambdaT}, satisfies conditions \eqref{cond}, that uniquely determine $T$ on $\B$. Indeed, using \eqref{conn}, we have that
\begin{align*}
y \leq z(\la) &\Rightarrow y-\epsilon < z(\la) \Rightarrow \La(y-\epsilon) > \la \Rightarrow q_\La^-=y\\
x \geq z(\la) &\Rightarrow \La(x) \leq \lambda \Rightarrow q_\La^-=x \\
x < z(\la) < y &\Rightarrow \La(z(\la)- \epsilon) > \la, \, \La(z(\la)) \leq \la \Rightarrow q_\La^-=z(\la). 
\end{align*}
Cases e), f), g) are similar and omitted for brevity. 
If instead $T$ is weakly upper semicontinuous, a similar argument can be applied to show that 
\begin{equation}\label{qLaright}
T(F)=q_\La^+(F), \text{ for each } F \in \B.
\end{equation} 

The third step of the proof is to show that $T$ is a $\La$-quantile with $\La$ given by \eqref{LambdaT} also on the set $\D$ of finitely supported distributions. Let $F \in \D$ be given by $F(x)=\sum_{i=1}^{n-1} \lambda_i \mathbf{1}_{[x_i, x_{i+1})} + \mathbf{1}_{[x_n, +\infty)}$, with $x_1 < \dots < x_n$ and $\lambda_1 < \dots < \lambda_{n-1}$.

There are two possibilities. If $T(F) \in (x_j, x_{j+1})$ for some $j=1, \dots, n-1$, then respectively from locality of $T$, the fact that $T$ coincides with $q_\La^-$ on $\B$, and locality of $q_\La^-$, it follows that 
\[
T(F)=T(B_{x_j,x_{j+1}}^{\la_j}) = q_\La^-(B_{x_j,x_{j+1}}^{\la_j}) = q_\La^-(F). 
\]
\noindent
If instead $T(F) =x_j$ for some $j=1, \dots, n$, there are two further subcases. If $j=1$ or $j=n$ the thesis follows immediately by locality as above. If $1 < j < n$, assume by contradiction that $q_\La^-(F) >T(F)$, and let $x_k \geq q_\La^-(F)$. From locality of $q_\La^-(F)$, it follows that $q_\La^-(F)=q_\La^-(\tilde{F})$, where $\tilde{F} \in \B$ is defined by
\[
\tilde{F}(x)=
\begin{cases}
0 &\text{ if } x < x_1 \\
\lambda_{j} &\text { if } x_{1} \leq x < x_k\\
1 &\text { if } x \geq x_k
\end{cases}
\]
and satisfies $\tilde{F}\leq_{st} F$. This leads to a contradiction since 
\[
x_j < q_\La^-(F) =  q_\La^-(\tilde{F}) =T(\tilde{F}) \leq T(F)=x_j. 
\]
A similar contradiction arises if $q_\La^-(F) <T(F)$, thus showing that $T(F)=q_\La^-(F)$ on $\mathcal{D}$.

As a fourth step, we prove that $T(F)=q_\La^-(F)$ for each $F \in \M_c$. If $F$ has bounded support, there exist a sequence of finitely supported distribution functions $F_n \in \mathcal{D}$ such that $F_n \downarrow F$ in the continuity points of $F$. By the previous step, $T(F_n)=q_\La^-(F_n)$, and by weakly lower semicontinuity of $T$ and $q_\La^-$ and Proposition \ref{lemma:FMP} it follows that 
\[
T(F)=\lim_{n \to +\infty}T(F_n) = \lim_{n \to + \infty} q_\La^-(F_n)=q_\La^-(F).
\]
\noindent 
Finally, the thesis can be obtained for a general $F \in \M$ by locality. Let $z=T(F)$ and let 
\[
G(t)=
\begin{cases}
0 &\text { if } t \leq z - 1\\
F(t) &\text { if } z-1 < t \leq z+1\\
1 &\text { otherwise}
\end{cases}
\]
Then $T(F)=T(G)=q_\La^-(G)=q_\La^-(F)$, that concludes the proof. 
\end{proof}

\section*{Acknowledgments}
We are grateful to the conference participants at the 11th ERCIM Conference in Pisa, the 9th General AMaMeF Conference in Paris, the 43nd AMASES Meeting in Perugia, and to the participants at the seminars at the University of Waterloo and University of Lyon I
for their useful comments and suggestions.

\bibliography{references}

\end{document}